\documentclass[11pt,a4paper]{article}
\usepackage[T1]{fontenc}
\usepackage{lmodern}
\usepackage[margin=24mm]{geometry}
\usepackage{amsmath,amsthm,amsfonts,amssymb,mathtools,booktabs,microtype}
\usepackage{tikz}
\usetikzlibrary{shapes.geometric}
\usepackage{enumitem}
\usepackage{hyperref}
\hypersetup{colorlinks=true,allcolors=black}
\setlist{itemsep=3pt,topsep=5pt}
\allowdisplaybreaks[2]
\numberwithin{equation}{section}
\newtheorem{theorem}{Theorem}[section]
\newtheorem{lemma}[theorem]{Lemma}
\newtheorem{proposition}[theorem]{Proposition}
\newtheorem{corollary}[theorem]{Corollary}
\theoremstyle{remark}
\newtheorem{remark}[theorem]{Remark}
\newcommand{\funding}[1]{#1}
\newcommand{\email}[1]{\href{mailto:#1}{\nolinkurl{#1}}}
\newenvironment{keywords}{\par\smallskip\noindent\textbf{Keywords. }}{\par\smallskip}
\newenvironment{MSCcodes}{\par\smallskip\noindent\textbf{MSC 2020. }}{\par\medskip}
\newcommand{\NN}{\mathbb Z_{\geq0}}
\newcommand{\KA}{\operatorname{KA}}
\newcommand{\CNA}{\operatorname{CNA}}
\newcommand{\KK}{\mathcal K}
\newcommand{\card}[1]{\lvert#1\rvert}
\newcommand{\ind}[1]{\mathbf1_{\{#1\}}}
\newcommand{\VerificationURL}{https://raw.githubusercontent.com/Yan-ll9/kloeve-optimality/0e0266c8baba3f583a8f8dedd427e33ed24bcf00/kloeve_verification_v3.zip}
\hypersetup{pdfauthor={Lilin Yan and Hongwei Zhao},
  pdftitle={Optimality of Klove Arrays within the Symmetric Klove--Mossige Class}}
\title{Optimality of Kl\o ve Arrays within the Symmetric Kl\o ve--Mossige Class\thanks{
\funding{This work was supported by the Key Research and Development
Program of Shaanxi Province, China, under grant 2026CY-YBXM-044.}}}
\author{Lilin Yan\thanks{School of Electronics and Information,
Northwestern Polytechnical University, Xi'an 710072, China
(\email{yanlilin@mail.nwpu.edu.cn}).}
\and Hongwei Zhao\thanks{Corresponding author. School of Electronics and
Information, Northwestern Polytechnical University, Xi'an 710072, China
(\href{mailto:hongvi_zhao@126.com}{\nolinkurl{hongvi_zhao@126.com}}).}}
\date{}
\begin{document}
\maketitle
\begin{abstract}
Rajam\"aki and Koivunen asked whether minimum-redundancy symmetric
Kl\o ve--Mossige arrays with contiguous sum co-arrays are always Kl\o ve
arrays. We give a computer-assisted proof that, at every fixed sensor
count, every maximizing sensor set belongs to the Kl\o ve class. The
argument classifies the overlaps between a generator and its shifted
reflection, then bounds the aperture of a hypothetical optimizer outside
the Kl\o ve class. Comparing these bounds with a classical Kl\o ve
construction settles all sensor counts at least 330. An exact integer
certificate covers the remaining counts from 2 through 329 and matches
every equality case to a Kl\o ve array as a complete set. Consequently,
optimization within the full symmetric Kl\o ve--Mossige class reduces to
the previously known search over Kl\o ve parameters. The theorem concerns
this specified class, rather than unrestricted sparse arrays or all
restricted additive bases.
\end{abstract}
\begin{keywords}
restricted additive basis, extremal set, sparse array, sum co-array,
computer-assisted proof
\end{keywords}
\begin{MSCcodes}
11B13, 05B10, 05-04
\end{MSCcodes}

\section{Problem and precise statement}
An array is a finite set $D\subset\NN$ with $0\in D$. Its aperture is
$L=\max D$, and its sum co-array is $D+D=\{d+e:d,e\in D\}$.
For a contiguous sum co-array, $D+D=[0,2L]$, where intervals throughout
this paper contain integers only. At a fixed sensor count $N=\card D$,
minimizing the redundancy $N(N+1)/(2(2L+1))$ is equivalent to maximizing $L$.
Rajam\"aki and Koivunen~\cite[Section V-B2]{RK2021} asked whether the
optimizers in the symmetric Kl\o ve--Mossige (S-KMA) class always belong
to its Kl\o ve-array (KA) subfamily. We prove that they do: at every
fixed sensor count, \emph{every} maximizing sensor set is a KA.
This identifies the optimizing sets, rather than only showing that a
KA attains the optimal value.

An S-KMA is formed from a generator and its shifted reflection.
Increasing the shift increases the aperture, but can change the number
of common points and hence the sensor count. Such an improvement need
not be a competitor in the fixed-$N$ problem. Our argument controls
this overlap and compares arrays at the same sensor count. The precise
definitions and theorem follow below.

\subsection{Relation to additive bases and prior constructions}
Adjoining zero to a positive additive 2-basis makes its usual range equal
to the endpoint of the initial interval in its sumset. Thus a set with
$N$ sensors and sumset $[0,2L]$ is a restricted additive 2-basis of length
$N-1$ in the convention of Kohonen~\cite{Kohonen2014}. Restricted bases
need not be symmetric; the symmetric S-KMA sets form a particular
parametric subclass. Kohonen's meet-in-the-middle search~\cite{Kohonen2014}
and its early-pruning refinement~\cite{Kohonen2015} restrict admissible
prefixes to find extremal bases in the larger class, through length 47
in the latter work. Our finite calculation instead enumerates prescribed
generators and shifts to verify the bounded remainder of an infinite,
class-specific theorem.

Kl\o ve~\cite{Klove1980} constructed symmetric 2-bases with asymptotic
range coefficient $6/23$. His nine-block construction and parity rounding
also give the quantitative lower bound used here; we spell out the change
of variables before Lemma~\ref{lem:lower}. Rajam\"aki and
Koivunen~\cite{RK2020} introduced this classical construction into array
processing as the KA and studied its redundancy and a constant-unit-spacing
subfamily. Their subsequent work~\cite[Section V-B]{RK2021} treats the
larger class obtained by symmetrizing a Kl\o ve--Mossige generator and
traces that generator to earlier additive-basis work. We take its explicit
generator definition as our starting point. Rajam\"aki and
Koivunen~\cite[Sections V-B2--V-B3]{RK2021} formulated the S-KMA
optimality question, analyzed the KA subfamily, and gave a search using
$O(N\log N)$ aperture evaluations within that subfamily. These
constructions, the lower-bound method, and the KA search are prior work.
The contribution of the present proof is the elimination of all
non-KA optimizing sets, including exceptional parameter regimes, together
with an exhaustive equality-set certificate for the finite remainder.

\subsection{Definitions and main theorem}
Here are explicit definitions, including the zero-parameter cases.
Let $x,y,z\in\NN$, with $x+y\geq1$. For $y\geq1$, put
\[
\begin{aligned}
 A=\CNA(x,y)&=A_1\cup S\cup A_2,\\
 A_1&=[0,x-1],\\
 S&=\{x+j(x+1):0\leq j<y\},\\
 A_2&=[y(x+1),y(x+1)+x-1].
\end{aligned}
\]
Empty intervals contribute no elements. For $y=0$, define $A=[0,x-1]$.
Write
\begin{align}
 M&=\max A, &a&=\card A, &Q&=\{ux:0\leq u\leq x\},\label{eq:base}\\
 P&=x^2+M+1, &T_z&=\bigcup_{i=0}^{z-1}(iP+Q), &G&=A\cup(2M+1+T_z).\label{eq:generator}
\end{align}
Thus $T_0=\varnothing$ and $Q=\{0\}$ when $x=0$.
With $m=\max G$, the S-KMA associated with a shift $\lambda\in\NN$ is
\begin{equation}\label{eq:symmetrization}
 D_{x,y,z}(\lambda)=G\cup(m+\lambda-G),\qquad L=m+\lambda.
\end{equation}
These are Definitions 11 and 12 of \cite{RK2021}, expressed as sets.
Define, provisionally also for $y=0$,
\begin{equation}\label{eq:ka}
 \KA(x,y,z)=A\cup(2M+1+T_z)\cup(A+2M+1+zP).
\end{equation}
The class used for the conclusion is
$\KK=\{\KA(x,y,z):x,z\in\NN,\ y\geq1\}$.
Let $L^*(N)$ be the maximum feasible S-KMA aperture and let
\[
 \ell(N)=\max\{\max K:K\in\KK,\ \card K=N\}.
\]
The constructions below show that both sets of competitors are nonempty
for every $N\geq2$. Their maxima exist: $G\subseteq D$ bounds all generator
parameters at a fixed $N$, and feasibility bounds $\lambda$ by $2m+1$.

\begin{theorem}\label{thm:main}
For every integer $N\geq2$, every feasible S-KMA of maximum aperture
with $N$ sensors belongs to $\KK$ as a sensor set. Consequently
$L^*(N)=\ell(N)$ for every $N\geq2$.
\end{theorem}
The set formulation matters. Different S-KMA parameter tuples may
represent the same array, so the theorem does not assert that every
optimal tuple uses a particular canonical shift. Figure~\ref{fig:set-equivalence}
illustrates this distinction with an optimal 12-sensor set. We first prove bounds
that settle all $N\geq330$, and then give an exhaustive finite certificate.

\begin{remark}[Scope and nonuniqueness]\label{rem:scope}
The class restriction is essential even among symmetric arrays. The
length-nine basis in \cite[Table 2]{Kohonen2014},
\[
 B=\{0,1,3,4,9,11,16,17,19,20\},
\]
satisfies $B=20-B$ and $B+B=[0,40]$, whereas $\ell(10)=19$.
The latter value follows by checking the nine triples with $y\geq1$
and $4x+2y+z(x+1)=10$ using \eqref{eq:NK}--\eqref{eq:LK}.
Thus the theorem does not extend to all symmetric restricted bases.

Within S-KMA, the maximizing set itself need not be unique. For $N=7$,
the two distinct sets
\begin{align*}
 \KA(0,3,1)&=\{0,1,2,5,8,9,10\},\\
 \KA(0,2,3)&=\{0,1,3,5,7,9,10\}
\end{align*}
both have sumset $[0,20]$ and attain the certified optimum $\ell(7)=10$
(Proposition~\ref{prop:finite}). This is different from the parameter
nonuniqueness in Figure~\ref{fig:set-equivalence}: several representations
may describe one set, and several different sets may be optimal.
\end{remark}

\begin{figure}[htbp]
\centering
\begin{tikzpicture}[x=0.13in,y=0.39in,line width=1.05pt,
 base/.style={circle,fill=black,inner sep=1.65pt},
 middle/.style={rectangle,fill=black,inner sep=1.55pt},
 rightbase/.style={regular polygon,regular polygon sides=3,fill=black,inner sep=1.5pt}]
\foreach \r in {0,1,2,3} {
 \draw[gray!60] (0,\r) -- (27,\r);
 \foreach \i in {0,...,27} \draw[gray!60] (\i,\r-0.04)--(\i,\r+0.04);
}
\node[anchor=east,font=\small] at (-1,3) {$G$};
\node[anchor=east,font=\small] at (-1,2) {$27-G$};
\node[anchor=east,font=\small] at (-1,1) {$D$};
\node[anchor=east,font=\small] at (-1,0) {$\KA(1,3,1)$};
\foreach \i in {0,1,3,5,6,13,14,21,22} \node[base] at (\i,3) {};
\foreach \i in {5,6,13,14,21,22,24,26,27} \node[base] at (\i,2) {};
\foreach \i in {0,1,3,5,6,13,14,21,22,24,26,27} \node[base] at (\i,1) {};
\foreach \i in {5,6,13,14,21,22} \draw[black] (\i,1) circle[radius=3.1pt];
\foreach \i in {0,1,3,5,6} \node[base] at (\i,0) {};
\foreach \i in {13,14} \node[middle] at (\i,0) {};
\foreach \i in {21,22,24,26,27} \node[rightbase] at (\i,0) {};
\foreach \i in {0,3,6,9,12,15,18,21,24,27}
 \node[below,font=\scriptsize] at (\i,-0.08) {\i};
\node[anchor=west,font=\scriptsize] at (28,3) {$9$ points};
\node[anchor=west,font=\scriptsize] at (28,2) {$9$ points};
\node[anchor=west,font=\scriptsize] at (28,1) {$12$ points};
\node[anchor=west,font=\scriptsize] at (28,0) {$12$ points};
\end{tikzpicture}
\caption{Set equality despite different parameters. For $(x,y,z,\lambda)
=(1,3,2,5)$, the generator has nine points and its reflection overlaps
it at six points (rings). Their union is the 12-sensor set $\KA(1,3,1)$,
with aperture 27. In the last row, circles mark the left base, squares
the middle block, and triangles the translated right base. All positions
are exact integers.}\label{fig:set-equivalence}
\end{figure}
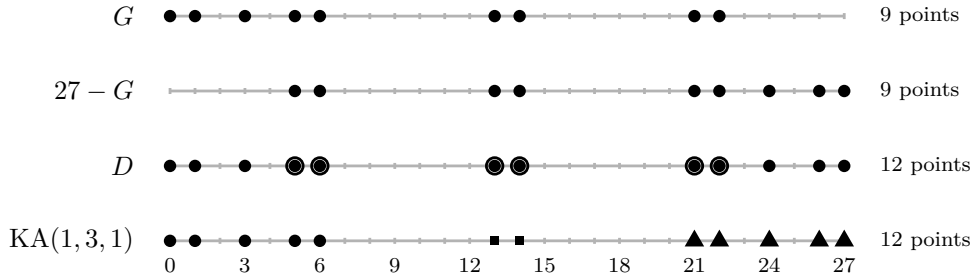

\subsection{Proof strategy and computational scope}
The proof separates structural identities, exclusion bounds and a finite
verification. This separation specifies the logical role of the computation.
\begin{enumerate}[label=(\roman*)]
\item We compute the generator's difference set and first missing sum.
These determine every feasible shift. Intersections with the shifted
reflection give the exact sensor count, including degenerate parameters.
\item We classify overlaps and bound the aperture of a hypothetical
maximizer outside $\KK$. Some cases are already KA sets; others are
strictly dominated at the same sensor count. The remaining cases satisfy
one of the two upper bounds in Proposition~\ref{prop:upper}.
\item A classical nine-block KA gives a uniform lower bound. For
$N\geq330$ it strictly exceeds both exclusion bounds, producing a
contradiction. This is the infinite part of the argument.
\item For $2\leq N\leq329$, bounded integer enumeration checks every
feasible shift and every equality set, giving Proposition~\ref{prop:finite}.
\end{enumerate}
The decisive comparison is in the linear term. The main exclusion bound
and the classical construction have the same quadratic coefficient,
but their difference, including the uniform rounding loss, is
\[
 \left(\frac{3N^2+18N-19}{23}-72\right)
 -\frac{3N^2+13N-30}{23}=\frac{5N-1645}{23}.
\]
It is positive for $N\geq330$. The other exclusion bound has the smaller
quadratic coefficient $1/8$ and is also exceeded in this range. Hence
only finitely many sensor counts remain. These exclusion bounds apply
to a hypothetical non-KA \emph{optimizer}, after domination alternatives
have been removed; they are not asserted for every non-KA array.

\section{Generator identities and feasible shifts}
\begin{lemma}\label{lem:cna}
For $y\geq1$,
$a=2x+y$ and $M=(x+1)(y+1)-2$; for $y=0$, $a=x$ and $M=x-1$.
In both cases
\[
 M-A=A,\qquad A+A=[0,2M],\qquad A+Q=[0,M+x^2].
\]
\end{lemma}
\begin{proof}
For $y\geq1$, the three displayed blocks of $A$ are disjoint and ordered;
reflection in $M/2$ interchanges $A_1,A_2$ and reverses $S$.
If $x=0$, then $A=[0,y-1]$ and the assertions are immediate.
Assume $x\geq1$. To represent $s\in[0,M]$ as a sum, write
$s=q(x+1)+r$ with $0\leq r\leq x$.
For $q=0$ or $q=y$, $s\in A$. For $1\leq q\leq y-1$, if $r=x$,
again $s\in S$. Otherwise
$q(x+1)-1\in S$ and $r+1\in A_1\cup\{x\}\subseteq A$ give the sum.
Reflection then supplies $[M,2M]$.

For the last identity, $A_1+Q=[0,x^2+x-1]$ and
$A_2+Q=[M-x+1,M+x^2]$. If these intervals do not meet or abut, take
$s\in[x^2+x,M-x]$ and write $s=q(x+1)+r$, $0\leq r\leq x$.
Then $x\leq q\leq y-1$. If $r\geq1$, set $t=x-r$, $j=q-t$;
then $0\leq t<x$, $0\leq j<y$ and $s=(x+j(x+1))+tx\in S+Q$.
If $r=0$, take $j=q-x$ and $t=x$ instead.
For $y=0$, all three identities follow from
$A=[0,x-1]$ and the abutting intervals $ux+A$, $0\leq u\leq x$.
\end{proof}

\begin{lemma}\label{lem:generator}
For $z\geq1$, the blocks of $T_z$ are disjoint and $T_z$ is symmetric.
Moreover, writing $B=2M+1+T_z$ and $n=\card G$,
\begin{align*}
 n&=a+z(x+1), &m&=M+zP,\\
 A+B&=[2M+1,m+M], &B-A&=[M+1,m],\\
 G-G&=[-m,m].
\end{align*}
For $z=0$, $n=a$, $m=M$ and the last identity still holds.
\end{lemma}
\begin{proof}
The gap between the end of one block and the start of the next is
$P-x^2=M+1>0$. Reflection in $( (z-1)P+x^2)/2$ reverses the
blocks and the points within each block. Also $A\cap B=\varnothing$.
Lemma~\ref{lem:cna} gives
\[
 A+T_z=\bigcup_{i=0}^{z-1}[iP,(i+1)P-1]=[0,zP-1].
\]
This proves the formulas for $A+B$ and, using $-A=A-M$, for $B-A$.
Finally $A-A=[-M,M]$; combine this with $B-A$ and its negative.
\end{proof}

\begin{lemma}[First hole]\label{lem:hole}
Let $H=\min\{s\in\NN:s\notin G+G\}$. For $z=0$, $H=2M+1$.
For $z\geq1$,
\begin{equation}\label{eq:hole}
 H=\begin{cases}
 2m+1,&x+y=1,\\
 m+M+2,&y=1,\ x\geq2,\\
 m+M+1,&\text{otherwise}.
 \end{cases}
\end{equation}
\end{lemma}
\begin{proof}
The $z=0$ case follows from Lemma~\ref{lem:cna}. If $x+y=1$, $G$ is
an integer interval. Otherwise $M\geq1$, and
$[0,m+M]\subseteq(A+A)\cup(A+B)$.
For $e=0$, or for $e=1$ when needed, the next potential sum
$m+M+1+e$ belongs to $B+B$ exactly when
\begin{equation}\label{eq:holecondition}
 (z-d)P-2M-1+e=vx,
 \qquad 0\leq d\leq2z-2,\quad 0\leq v\leq2x.
\end{equation}
For $d\geq z$, the left side is negative. For $d\leq z-3$, it is
at least $3x^2+M+2>2x^2$. For $d=z-2$ it is $2x^2+1+e$.
The only possible index is $d=z-1$, giving $x^2-M+e=vx$.

Suppose first $x,y\geq1$ and $e=0$. This requires $x\mid M$ and
$M\leq x^2$. Since $M=x(y+1)+y-1$, if $y\geq2$ then divisibility
forces $y\geq x+1$, contradicting $M\leq x^2$.
If $y=1$, then $M=2x$ and the condition holds precisely when $x\geq2$,
with $v=x-2$. In this exceptional case, the $e=1$ expression
$(x-1)^2$ is not divisible by $x$, so the first missing sum is the next one.
For $y=0,x\geq2$, the $e=0$ expression is $x^2-x+1$, not divisible by $x$.
For $x=0,y\geq2$, it is $-M<0$. This proves every case.
\end{proof}

\begin{lemma}[Feasibility]\label{lem:feasibility}
The sum co-array of $D_{x,y,z}(\lambda)$ is contiguous if and only if
$0\leq\lambda\leq H$.
\end{lemma}
\begin{proof}
Put $L=m+\lambda$. Directly from the definition,
\begin{equation}\label{eq:generic}
 D+D=(G+G)\cup\bigl(L+(G-G)\bigr)\cup\bigl(2L-(G+G)\bigr).
\end{equation}
The set is symmetric about $L$. Below $\lambda$, only $G+G$ can contribute,
because the other two terms begin at $\lambda$ and $2\lambda$ respectively.
The middle term covers $[\lambda,L]$ by Lemma~\ref{lem:generator}.
Thus $[0,L]\subseteq D+D$ precisely when $[0,\lambda-1]\subseteq G+G$.
\end{proof}

For the overlap calculations define
\[
 q(t)=\card{Q\cap(Q+t)}=
 \begin{cases}
 x+1-|t|/x,&x\geq1,\ x\mid t,\ |t|\leq x^2,\\
 \ind{t=0},&x=0,\\
 0,&\text{otherwise}.
 \end{cases}
\]
\begin{lemma}[Sensor count]\label{lem:count}
Let $\rho(\lambda)=\card{G\cap(m+\lambda-G)}$, so that
$N=2n-\rho(\lambda)$. For $z\geq1$,
\begin{align}
 \rho(\lambda)&=2r(\lambda)+
 \sum_{d=-(z-1)}^{z-1}(z-|d|)q(\lambda-2M-1+dP),\label{eq:rho}\\
 r(\lambda)&=\card{A\cap(\lambda+Q)}.\notag
\end{align}
In particular $r(\lambda)=0$ for $\lambda>M$.
For $z=0$, $\rho(\lambda)=\card{A\cap(M+\lambda-A)}$.
\end{lemma}
\begin{proof}
The $A$--$A$ contribution is empty for $z\geq1$ because $zP>M$.
Symmetry of $T_z$ gives $m+\lambda-B=\lambda+T_z$.
Only its first block can intersect $A$, since $P>M$.
The two cross contributions have equal size by reflection and total $2r$.
For the $B$--$B$ contribution, group pairs of blocks by their index
difference $d$; there are $z-|d|$ such pairs, each contributing the indicated $q$.
\end{proof}

\begin{lemma}[Kl\o ve identities]\label{lem:ka}
For $z\geq1$ and $0\leq k\leq z$,
\begin{equation}\label{eq:canonical}
 D_{x,y,z}(2M+1+kP)=\KA(x,y,z+k).
\end{equation}
For $y\geq1$ and $t\geq0$, the sensor count and aperture of $\KA(x,y,t)$ are
\begin{align}
 N_K(t)&=4x+2y+t(x+1),\label{eq:NK}\\
 L_K(t)&=3M+1+tP=(x+1)\bigl(t(x+y)+3y+3\bigr)-5.\label{eq:LK}
\end{align}
Each such KA is feasible. The following set identities also hold:
\begin{align}
 \KA(0,y,t)&=\CNA(y-1,t+2) &&(y\geq1),\label{eq:cnaka}\\
 \KA(x,0,t)&=\CNA(x-1,t(x+1)+2)
              =\KA(0,x,t(x+1)) &&(x\geq1).\label{eq:zero}
\end{align}
Every interval $[0,k-1]$ with $k\geq2$ equals $\KA(0,1,k-2)$.
Finally, if $y,z\geq1$, $\lambda=M-x^2\geq0$ and $r(\lambda)=x+1$, then
\begin{equation}\label{eq:absorb}
 D_{x,y,z}(\lambda)=\KA(x,y,z-1).
\end{equation}
\end{lemma}
\begin{proof}
In \eqref{eq:canonical}, the reflected $A$ is $A+2M+1+(z+k)P$,
and the reflected $B$ comprises middle blocks $k,\ldots,z+k-1$.
Their union with blocks $0,\ldots,z-1$ gives exactly \eqref{eq:ka}.
The three parts of \eqref{eq:ka} are disjoint, giving \eqref{eq:NK}--\eqref{eq:LK}.
Here $P=(x+1)(x+y)$ for $y\geq1$.
The shift $2M+1$ is feasible by Lemmas~\ref{lem:hole} and \ref{lem:feasibility};
for $t=0$ the same identity holds directly.

When $x=0$, the middle points are $(i+2)y-1$ and the two end blocks
are $[0,y-1]$ and $[(t+2)y-1,(t+3)y-2]$. Regrouping their end points
gives \eqref{eq:cnaka}. When $y=0$, $T_t=x[0,t(x+1)-1]$;
the same regrouping gives \eqref{eq:zero}. The interval identity follows
by taking $x=0,y=1$.

For \eqref{eq:absorb}, symmetry of $A$ and $Q$ gives
$r(M-x^2)=\card{A\cap Q}$, so $Q\subseteq A$ and $\lambda+Q\subseteq A$.
All later reflected middle blocks coincide with the preceding original
blocks since $\lambda+P=2M+1$.
The last original block is absorbed by the reflected $A$, which is
$A+2M+1+(z-1)P$. The remaining union is \eqref{eq:absorb}.
\end{proof}

\begin{lemma}[Removing zero overlap]\label{lem:remove}
At $\lambda=H$, $\rho(H)=0$. A feasible shift $\lambda<H$ with
$\rho(\lambda)=0$ cannot be optimal at its sensor count.
\end{lemma}
\begin{proof}
For $z\geq1$, $H>M$ and each offset in \eqref{eq:rho} at $H$ is
at least $P>x^2$, by \eqref{eq:hole}; this also holds for the interval cases.
For $z=0$, $M+H>2M$. Thus both $D(\lambda)$ and $D(H)$ have $2n$
sensors, and the latter has strictly larger aperture.
\end{proof}
This lemma compares with a feasible S-KMA. It does not assume that $D(H)$ is a KA.
We will also use the elementary prefix bound, for $0\leq\lambda\leq M$,
\begin{equation}\label{eq:prefix}
 N\geq n+\card{A\cap[0,\lambda-1]}.
\end{equation}
Indeed each $u\in A$ with $u<\lambda$ produces a distinct reflected
sensor $m+\lambda-u>m$, outside $G$.

\section{Relaxations of the KA aperture}
Eliminating $t$ from \eqref{eq:NK}--\eqref{eq:LK}, at sensor count $V$,
gives the polynomial identity
\begin{equation}\label{eq:relaxidentity}
 L_K=(x+y)(V+3)-3xy-4x^2-2y^2-2.
\end{equation}
For fixed $x$, completing the square in $y$ yields
\begin{align}
 U_x(V)&=\frac{V^2+2Vx+6V-23x^2+6x-7}{8},\label{eq:Ux}\\
 U_x(V)-L_K&=2\left(y-\frac{V-3x+3}{4}\right)^2.\label{eq:square}
\end{align}
These are the fixed-$x$ relaxation of \cite[Theorem 4]{RK2021},
and the identities above give a direct derivation.
Maximizing in real $x$ gives
\begin{equation}\label{eq:U}
 U(V)=\frac{3V^2+18V-19}{23}.
\end{equation}
For $x\geq0$, $U_x(V)$ increases with $V\geq0$. Define
\begin{equation}\label{eq:Phi}
 \Phi(x,N)=U_x(N-1)+x,\qquad
 F(N)=\max_{x\in\mathbb R}\Phi(x,N)=\frac{3N^2+13N-30}{23}.
\end{equation}
The maximum is attained at $x=(N+6)/23$.
Two useful identities are
\begin{align}
 8\bigl(U_x(N-1)-U_x(N-s)\bigr)
   &=(s-1)(2N-s+2x+5),\label{eq:difference}\\
 U(N)-F(N)&=\frac{5N+11}{23}.\label{eq:gap}
\end{align}

\section{Upper bounds for a non-Kl\o ve optimizer}
In this section it suffices to treat $N\geq6$. An array is called
\emph{KA-dominated} if a KA with the same sensor count has strictly
larger aperture. Such an array cannot be an optimizer.

The following partition makes the coverage of the argument explicit.
First, the cases $z=0$, $z\geq1,x=0$, and $z\geq1,x\geq1,y=0$
are handled by Lemmas~\ref{lem:zzero}, \ref{lem:xzero}, and
\ref{lem:yzero}, respectively. These are all degenerate generators.
For the remaining case $x,y,z\geq1$, Lemma~\ref{lem:remove} excludes
zero overlap below $H$. The endpoint $H$ is a KA by
\eqref{eq:canonical}, except when $y=1,x\geq2$, which is treated
in Lemma~\ref{lem:high}. Positive overlap is split at $\lambda=M$:
Lemma~\ref{lem:low} covers $0\leq\lambda\leq M$, and
Lemma~\ref{lem:high} covers $M<\lambda\leq H$.
Every branch yields KA membership, strict domination, or one of the
two bounds collected in Proposition~\ref{prop:upper}.

\subsection{Positive overlap above the CNA aperture}
\begin{lemma}\label{lem:high}
Suppose $x,y,z\geq1$, $M<\lambda\leq H$ and $\rho(\lambda)>0$.
Then $D(\lambda)\in\KK$ or $L\leq\Phi(x,N)$.
At the exceptional endpoint $H=m+M+2$ with $y=1,x\geq2$,
the bound $L\leq\Phi(x,N)$ holds as well.
\end{lemma}
\begin{proof}
Put $\delta=\lambda-2M-1$. Since $\lambda>M$, no index $d\geq1$
can contribute to \eqref{eq:rho}. There is at most one active index.
Indeed $P>x^2$ permits at most two, necessarily adjacent.
If both were active, both offsets would be divisible by $x$, so $x\mid P$.
Since $P=x^2+x(y+1)+y$, this forces $x\mid y$, hence $y\geq x$.
It follows that $P\geq2x^2+2x>2x^2$, whereas two offsets in $[-x^2,x^2]$
cannot be this far apart. This is a contradiction.

Let $d\in[-(z-1),0]$ be the active index. Put
$w=z+d\geq1$, $t=z-d\geq1$, and $\delta+dP=ux$ with $|u|\leq x$.
If $u=0$, Lemma~\ref{lem:ka} applies. Otherwise
\[
 N=N_K(t)+s,\qquad L=L_K(t)+ux,\qquad s=w|u|\geq1.
\]
For $u<0$, monotonicity gives $L\leq U_x(N-1)\leq\Phi(x,N)$.
For $u>0$, \eqref{eq:difference} reduces the desired result to
\begin{equation}\label{eq:highineq}
 8x(u-1)\leq(s-1)(2N-s+2x+5).
\end{equation}
Here $N-s=N_K(t)\geq5x+3$, so the second factor is at least $12x+12$;
the first factor is at least $u-1$. This proves \eqref{eq:highineq}.

At the exceptional endpoint, $\rho(H)=0$, $N=N_K(2z)$ and
$L=L_K(2z)+1$. Put $v=N-3x-1=x+1+2z(x+1)\geq9$.
Using \eqref{eq:square} with $y=1$ gives directly
\[
 8\Phi(x,N)-8L=(v-5)(v+3)>0.
\]
\end{proof}

\subsection{Overlap at or below the CNA aperture}
\begin{lemma}\label{lem:low}
Suppose $x,y,z\geq1$ and $0\leq\lambda\leq M$.
Then $D(\lambda)$ belongs to $\KK$, is KA-dominated, or has
$L\leq\Phi(x,N)$.
\end{lemma}
\begin{proof}
Put $t=z-1$, $\tau=\lambda-(M-x^2)$, $q_0=q(\lambda-2M-1)$,
$q_1=q(\tau)$ and $r=r(\lambda)$.
Only the indices $d=0,1$ can contribute in \eqref{eq:rho}.
Thus, with $s=N-N_K(t)$,
\begin{align}
 L&=L_K(t)+\tau,\notag\\
 s&=2(x+1-r)+(z-1)(x+1-q_1)-zq_0.\label{eq:s}
\end{align}
Here $s$ is the sensor-count difference from the comparison array
$\KA(x,y,t)$, and $\tau$ is its aperture difference. Controlling $s$
in terms of $\tau$ allows a comparison at the actual sensor count $N$.
We need three elementary facts. First,
$q_0+q_1\leq x+1$ because the two translated copies of $Q$ involved
are disjoint ($P>x^2$).
Second, $r\leq3$: the progression $\lambda+ux$, $0\leq u\leq x$,
meets each end interval of $A$ at most once, and meets $S$ at most once
because its residues modulo $x+1$ are all distinct.
Third, if $q_0>0$, then $s\geq x$. In fact
\[
 r\leq\left\lfloor\frac{M-\lambda}{x}\right\rfloor+1,
 \qquad q_0=x+1-\frac{2M+1-\lambda}{x},
\]
so \eqref{eq:s} and the first fact imply
\[
 s\geq2(x+1-r)-q_0
   \geq x-1+\frac{\lambda+1}{x}>x-1.
\]
The integrality of $s$ proves the claim. These facts also show $s\geq0$
in every case, since all terms of \eqref{eq:s} are nonnegative when $q_0=0$.

If $\tau\leq0$ and $s\geq1$, the result follows from
$L\leq U_x(N-s)\leq U_x(N-1)$.
If $s=0$, the preceding facts force $q_0=0$ and $r=x+1$.
When $\tau<0$, the array is KA-dominated by $\KA(x,y,t)$.
When $\tau=0$, identity \eqref{eq:absorb} says it is that KA itself.

If $1\leq\tau\leq x$, equality $s=0$ would imply $r=x+1$, requiring
$\lambda+x^2=M+\tau\in A$, which is impossible. Hence $s\geq1$ and
$L\leq U_x(N-1)+x$.

It remains to treat $x+1\leq\tau\leq x^2$, so $x\geq2$.
We claim $s\geq x$. This is already proved if $q_0>0$.
If $q_0=0$, membership $\lambda+ux\in A$ requires
$(x-u)x\geq\tau\geq x+1$, so $u\leq x-2$.
Consequently $r\leq1$ for $x=2$, $r\leq2$ for $x=3$, and $r\leq3$
for $x\geq4$. Equation~\eqref{eq:s} gives $s\geq4,4,2x-4$ respectively,
always at least $x$.
Now $N-s=N_K(t)\geq4x+2$, whence
\[
 (s-1)(2N-s+2x+5)\geq(x-1)(11x+9)
     \geq8x(x-1)\geq8(\tau-x).
\]
Apply \eqref{eq:difference} to $L\leq U_x(N-s)+\tau$.
\end{proof}

\subsection{Degenerate generators}
\begin{lemma}\label{lem:xzero}
If $x=0$, $y,z\geq1$, every feasible $D(\lambda)$ is a KA or is
KA-dominated.
\end{lemma}
\begin{proof}
Here $A=[0,y-1]$, $M=y-1$, $P=y$, $m=(z+1)y-1$, and $n=y+z$.
For $\lambda>M$, positive overlap requires
$\lambda=2M+1+jP$, $0\leq j<z$, hence gives a KA.
The endpoint $H$ is also a KA: if $y\geq2$, use \eqref{eq:canonical}
with $j=z$; if $y=1$, then $G=[0,m]$ and
$D(H)=\KA(0,m+1,0)$. Every other zero-overlap shift is dominated by this endpoint.

For $\lambda\leq M$, \eqref{eq:rho} gives $\rho=2+(z-1)\ind{\lambda=M}$.
At $\lambda=M$, the set is
\[
 [0,y-1]\ \cup\ \{(i+1)y-1:1\leq i\leq z\}
 \ \cup\ \bigl([0,y-1]+(z+1)y-1\bigr),
\]
which equals $\CNA(y-1,z+1)=\KA(0,y,z-1)$.
For $\lambda\leq y-2$, we have $N=2y+2z-2$ and
$L\leq(z+2)y-3$. The KA $\KA(0,y,2z-2)$ has this same $N$ and
aperture $(2z+1)y-2$, strictly larger since $z\geq1$.
\end{proof}

\begin{lemma}\label{lem:yzero}
If $y=0$, $x,z\geq1$, every feasible $D(\lambda)$ is a KA,
is KA-dominated, or satisfies $L\leq(N+3)^2/8-2$.
\end{lemma}
\begin{proof}
If $x=1$, then $G=[0,2z]$. For $\lambda\leq m+1$ the reflected
union is an interval, hence a KA. The endpoint $H=2m+1$ gives
$\KA(0,m+1,0)$; all remaining shifts have zero overlap and are dominated.
Assume $x\geq2$, and put $K=z(x+1)$. Then
\[
 T_z=x[0,K-1],\quad M=x-1,\quad m=x(K+1)-1,\quad n=x+K.
\]
The cross overlap is $r=\ind{\lambda\leq x-1}$, and the middle overlap
is $K-|u|$ if $\lambda-2x+1=ux$, $|u|\leq K-1$, and is zero otherwise.
After Lemma~\ref{lem:remove}, the possibilities are as follows.
For $\lambda=2x-1+ux$, $0\leq u\leq K-1$,
\[
 N=2x+K+u,\qquad L=x(N-2x+3)-2\leq\frac{(N+3)^2}{8}-2.
\]
At $\lambda=x-1$ the same expression for $L$ holds, with $N=2x+K-1$.
For $\lambda\leq x-2$, $N=2x+2K-2$ and
\begin{align*}
 L&\leq xK+2x-3
   =\frac{x(N-2x+2)}2+2x-3\\
  &\leq\frac{(N+2)^2}{16}+N-3
   \leq\frac{(N+3)^2}{8}-2.
\end{align*}
Here $2x\leq N$, and the last difference is $((N-4)^2+14)/16>0$.
Finally $\lambda=H=2M+1+zP$ gives $\KA(x,0,2z)$,
which belongs to $\KK$ by \eqref{eq:zero}.
\end{proof}

\begin{lemma}\label{lem:zzero}
If $z=0$ and $N\geq6$, every feasible $D(\lambda)$ is a KA,
is KA-dominated, or satisfies $L\leq(N+3)^2/8-2$.
\end{lemma}
\begin{proof}
When $y=0$, the generator $[0,x-1]$ is identically the generator
with parameters $(0,x,0)$, so we may assume $y\geq1$.
At $\lambda=0$, $D=A=\KA(0,x+1,y-2)$ if $y\geq2$;
if $y=1$, $A=[0,2x]$ is an interval and the sensor restriction excludes
the singleton. At $\lambda=2M+1$, $D=\KA(x,y,0)$.
Every $M<\lambda<2M+1$ is dominated by that endpoint.

For $1\leq\lambda\leq M$, let $c=\card{A\cap[0,\lambda-1]}$.
The prefix bound gives $2x+y+c\leq N$ and $L=M+\lambda$.
If $\lambda\leq x$, then $c\geq\lambda$ and
\begin{align*}
 8L&\leq(N-\lambda+3)^2-16+8\lambda\\
   &=(N+3)^2-16-\lambda(2N-2-\lambda)
     \leq(N+3)^2-16.
\end{align*}
We used $(x+1)(N-\lambda+1-2x)\leq(N-\lambda+3)^2/8$.
If $\lambda\geq x+1$, the dense prefix and the sparse progression give
\[
 c\geq x+\left\lfloor\frac{\lambda-1-x}{x+1}\right\rfloor+1
   \geq x+\frac{\lambda-1-x}{x+1}.
\]
Consequently $\lambda\leq(N-3x-y+1)(x+1)$ and
\[
 L\leq(x+1)(N-3x+2)-2
    \leq\frac{(N+5)^2}{12}-2
    \leq\frac{(N+3)^2}{8}-2.
\]
The last comparison holds for $N\geq6$ because its difference is
$(N^2-2N-23)/24>0$.
\end{proof}

\begin{proposition}[Optimizer bound]\label{prop:upper}
Let $N\geq6$, and suppose an optimal feasible S-KMA $D$ with $N$
sensors does not belong to $\KK$. Then
\[
 \max D\leq F(N)\qquad\text{or}\qquad
 \max D\leq\frac{(N+3)^2}{8}-2.
\]
\end{proposition}
\begin{proof}
Lemma~\ref{lem:remove} removes all zero-overlap shifts below $H$.
For $x,y,z\geq1$, Lemmas~\ref{lem:high} and \ref{lem:low} handle
the positive-overlap shifts. The endpoint $H$ is a KA by
\eqref{eq:canonical}, except for the endpoint covered explicitly by
Lemma~\ref{lem:high}. Maximize $\Phi(x,N)$ using \eqref{eq:Phi}.
The other families are covered by Lemmas~\ref{lem:xzero}--\ref{lem:zzero}.
KA membership and all domination alternatives are excluded by the hypothesis.
\end{proof}

\section{A uniform lower bound and the infinite range}\label{sec:infinite}
The construction in the next lemma is classical. In the notation on
page 179 of Kl\o ve~\cite{Klove1980}, put $k=N-1$, $X=x+1$, $Y=y$
and $l=9$, adjoining zero to his positive basis. His $T(0)$ is our $A$,
his block spacing is $P$, and his basis with zero adjoined is
$\KA(x,y,9)$. His choice $X$ nearest $(k+4)/23$ with parity opposite
to $k$ becomes the choice of $x$ below. In particular, his exact
range formula, divided by two, gives
\[
 \frac12\left(\frac{6k^2+48k-134}{23}-138\theta^2\right)
 =U(N)-3-69\theta^2,
 \qquad \theta=x-\frac{N-20}{23}.
\]
For our range $N\geq43$, the additional hypothesis $Y\geq X\geq1$
in that construction holds: $x\geq0$ and
$y-x-1\geq(8N-298)/46>0$. We include the short direct calculation
so that all constants used in the comparison can be checked here.

\begin{lemma}[Kl\o ve's construction, reparameterized]\label{lem:lower}
For every integer $N\geq43$, $\ell(N)\geq U(N)-72$.
\end{lemma}
\begin{proof}
Set $z=9$ and choose $x\equiv N+1\pmod2$ with
$|x-(N-20)/23|\leq1$. Such an integer is one of the two consecutive
integers starting at $\lfloor(N-20)/23\rfloor$.
Put $y=(N-9-13x)/2$. The parity ensures $y$ is an integer.
Since $N\geq43$, $x\geq0$; moreover
\[
 x\leq\frac{N+3}{23}\leq\frac{N-11}{13},
\]
the second inequality already holding for $N\geq30$. Thus $y\geq1$.
The KA has $N$ sensors. Substitution into \eqref{eq:LK} gives
\begin{align*}
 L_K&=6Nx+6N-69x^2-120x-56\\
    &=U(N)-3-69\left(x-\frac{N-20}{23}\right)^2
      \geq U(N)-72.
\end{align*}
\end{proof}

\begin{proof}[Proof of Theorem~\ref{thm:main} for $N\geq330$]
Suppose an optimizer $D$ is not a KA. By Proposition~\ref{prop:upper},
it satisfies one of the two upper bounds. But for $N\geq330$,
\[
 U(N)-F(N)=\frac{5N+11}{23}>72.
\]
Also
\[
 U(N)-72-\frac{(N+3)^2}{8}
   =\frac{N^2+6N-13607}{184}>0
\]
for all $N\geq114$. In either case
$\max D<U(N)-72\leq\ell(N)\leq L^*(N)=\max D$,
a contradiction. The finite certificate below supplies the remaining sizes.
\end{proof}

\section{The finite range}\label{sec:finite}
The remaining sensor counts are covered by the following
computer-assisted proposition.
\begin{proposition}[Finite-range certificate]\label{prop:finite}
For every integer $2\leq N\leq329$ and every feasible S-KMA $D$ with
$\card D=N$, one has $\max D\leq\ell(N)$. If equality holds, then
$D\in\KK$ as a sensor set.
\end{proposition}
We first specify a finite search space containing every relevant array,
then justify the checks that certify the proposition. The execution
record and independent witness checks are reported separately.

\subsection{A complete finite search space}\label{sec:coverage}
Set $B=329$. Since $G\subseteq D$, every relevant generator satisfies
\begin{equation}\label{eq:finitebounds}
 a+z(x+1)\leq B,\qquad
 a=\begin{cases}2x+y,&y\geq1,\\x,&y=0.\end{cases}
\end{equation}
Thus it suffices to visit $x,y\in[0,B]$, excluding $x=y=0$,
retain those pairs with $a\leq B$, and take every
$z\in[0,\lfloor(B-a)/(x+1)\rfloor]$. Each eligible generator is
constructed as an explicit set.

The KA comparison sets are enumerated with $y\geq1$ under the same
generator bound; their cardinalities and apertures are calculated from
the constructed sets. Every KA with at most $B$ sensors is included
because its generator is a subset. Retaining all KA sets with maximum
aperture at each $N$ therefore determines $\ell(N)$ and a complete
list of maximizing KA representatives. Duplicate representations are
permitted; the list is not capped.

For each generator $G$, with $m=\max G$, the program forms the ordered
representation counts
\[
 c(s)=\card{\{(g,h)\in G^2:g+h=s\}},
\]
and directly checks that $G-G=[-m,m]$. Let $H$ be the first zero of
$c$, with $H=2m+1$ if all sums through $2m$ are present.
Identity~\eqref{eq:generic} then shows that the feasible shifts are
exactly the integers $0\leq\lambda\leq H$: below $\lambda$ the sumset
must be supplied by $G+G$, and the difference term supplies
$[\lambda,m+\lambda]$. Reflection supplies the other half.
Consequently the search omits no feasible shift.

\subsection{Checks and their correctness}\label{sec:finitechecks}
At a shift $\lambda$, put $L=m+\lambda$. The map
\[
 \{(g,h)\in G^2:g+h=L\}\longrightarrow G\cap(L-G),
 \qquad (g,h)\longmapsto g,
\]
is a bijection: for each $g$ in the intersection, the unique partner
is $h=L-g\in G$. In particular, the overlap is the \emph{ordered}
representation count $c(L)$, including a diagonal representation once.
Hence the sensor count is exactly
\begin{equation}\label{eq:finitecount}
 N=2\card G-c(L),
\end{equation}
where $c(L)=0$ when $L>2m$.

The following checks are made for every enumerated generator and shift.
\begin{enumerate}
\item Compute $N$ by \eqref{eq:finitecount}. If $N\notin[2,B]$, continue
to the next shift.
\item Reject the run if $L>\ell(N)$.
\item If $L=\ell(N)$, explicitly form and sort $G\cup(L-G)$, remove
duplicates, and check its cardinality and aperture. Compare its entire
coordinate vector with the retained maximizing KA sets at that $N$.
Reject the run if no identical set is found.
\end{enumerate}
In addition, every $N\in[2,B]$ must have at least one checked equality
case. These checks imply the proposition upon successful completion:
the search space covers every feasible array, the second check rules
out a larger aperture, and the third certifies KA membership for every
array attaining equality. Agreement of optimal values alone would not
suffice for the last conclusion.

\subsection{Execution and independent checks}\label{sec:execution}
The direct implementation, \path{direct_full_certificate.c}, uses integer arithmetic,
visits every feasible shift, and does not use the family-specific
first-hole formula \eqref{eq:hole}, overlap formula \eqref{eq:rho}, or
zero-overlap reduction to determine feasibility or cardinality.
The recorded run covers $246\,028$ generators and checks $1\,783$
equality occurrences. Its output gives the aperture and equality count
for every $N=2,\ldots,329$.
Arithmetic bounds, implementation details and full run totals are in
Appendix~\ref{app:computation}.

\begin{proof}[Computer-assisted proof of Proposition~\ref{prop:finite}]
Section~\ref{sec:coverage} gives a finite enumeration containing every
feasible S-KMA with $2\leq N\leq329$, together with every relevant KA
comparison set. Section~\ref{sec:finitechecks} establishes that successful
completion of the checks implies both conclusions of the proposition.
The direct integer run completed all checks successfully, as recorded
above and in the accompanying computational supplement.
\end{proof}

The direct implementation shares its constructors, sorting, baseline
storage and equality matcher with the formula-based C verifier.
Their complete per-$N$ aperture and equality-count vectors agree, but
these implementations are not wholly independent. A separately written
Python bitset verifier, sharing no C implementation, also checked
every potentially optimal sum co-array through $N=44$:
$1\,039\,278$ shifts and $11\,046$ direct $D+D$ checks passed.

To check the common equality matcher independently over the full finite
range, all $1\,783$ successful occurrences were exported as tuples
$(N,x,y,z,\lambda,L)$. A separate Python implementation reconstructs
each generator and its reflected union, checks cardinality and aperture,
and forms its sumset by arbitrary-precision integer bit shifts.
It constructs all $100\,411$ eligible KA triples and compares each
witness with the maximizing KA sets at the same $N$.
All witnesses pass and represent $401$ distinct sensor sets; their
per-$N$ counts agree with the full direct certificate. The output
records a matching KA triple and a set digest for each occurrence.
This independently validates the supplied witnesses, not completeness
of their enumeration, which rests on the direct search and its coverage
argument.

A separately written Python implementation covers all $2\leq N\leq329$.
It enumerates generators by cardinality and reconstructs sensor sets
without the C constructors or matcher. Ordered sum counts come from
squaring $\sum_{g\in G}b^g$: each coefficient is at most $\card G\leq329$,
so base $b=2^{16}$ prevents carry between coefficients. Exact integer
arrays then test every feasible shift. Only after the search are stored
results compared: every aperture, equality count and equality tuple agrees.
This program shares helpers with the Python generator audit, not with
the C programs or the witness checker.

The article-only supplement, \texttt{verification-v3}, contains the exact
verification source snapshot, a SHA-256 manifest, recorded outputs and a
contents index. It is publicly available at
\url{https://github.com/Yan-ll9/kloeve-optimality}.
The \href{\VerificationURL}{fixed-commit ZIP archive} identifies the exact
snapshot used in this article, independently of later repository updates.
Reproduction commands are in Appendix~\ref{app:computation}.
The computation is a finite component of the proof; the analytic
argument is required for all larger sensor counts.

\begin{proof}[Completion of the proof of Theorem~\ref{thm:main}]
Section~\ref{sec:infinite} proves KA membership for every optimizer
when $N\geq330$. Proposition~\ref{prop:finite} gives the same conclusion
for $2\leq N\leq329$. Since KA sets are feasible S-KMAs, their optimum
$\ell(N)$ equals $L^*(N)$ in both ranges.
\end{proof}

\section{Consequence for exact optimization}\label{sec:algorithm}
\begin{corollary}\label{cor:search}
For any $N\geq2$, an optimal contiguous-sum S-KMA and its aperture
can be found with $O(N\log N)$ exact integer objective evaluations.
Keeping all maximizing parameter triples represents every optimal sensor
set, possibly more than once.
\end{corollary}
\begin{proof}
By Theorem~\ref{thm:main}, it suffices to search the KA class. Enumerate
\[
 0\leq x\leq\left\lfloor\frac{N-2}{4}\right\rfloor,
 \qquad
 0\leq z\leq\left\lfloor\frac{N-4x-2}{x+1}\right\rfloor,
 \qquad y=\frac{N-4x-z(x+1)}2,
\]
retaining precisely the integral values of $y$. These bounds enforce
$y\geq1$. Evaluate $3M+1+zP$ for each retained triple. Every eligible
KA triple occurs, so every maximizing KA set is represented. If
$q=\lfloor(N-2)/4\rfloor$, the number of tested pairs is at most
\[
 \sum_{x=0}^{q}\left(1+\frac{N}{x+1}\right)
 \leq q+1+N\bigl(1+\log(q+1)\bigr)=O(N\log N).
\]
A maximizing triple constructs a set in $O(N)$ further operations.
The bound counts arithmetic evaluations, not bit operations or the cost
of materializing and deduplicating all optimal sets.
\end{proof}
This is the known KA grid search of \cite[Proposition 3 and
Algorithm 1]{RK2021}, with $y\geq1$ enforced explicitly. The new
consequence is its guarantee over the entire feasible S-KMA class; no
new search algorithm is claimed.

\section{Concluding remarks}
The fixed-sensor-count problem is governed by overlap: changing the
reflection shift affects the aperture and can also change the number
of sensors. Classifying that tradeoff either identifies a KA,
gives strict domination at the same sensor count, or yields an
exclusion bound for a hypothetical optimizer. The linear separation
from a classical construction then reduces the proof to a finite
verification. Checking equality as sets completes the reduction of
S-KMA optimization to the existing KA search.

The threshold 330 is sufficient for this comparison, rather than a
claimed structural transition. The theorem does not by itself count
distinct maximizing sets: the KA search lists all
maximizing triples, but several triples can encode one set and several
sets can maximize the aperture at a given $N$. Understanding how the
number of such sets varies with $N$ is a natural question beyond the
class-membership result.

\appendix
\section{Alignment with the reference definitions}
Our definitions follow the arXiv v2 version of \cite{RK2021}.
We derive the identities used in the proof directly because two displayed
formulas in that version need care at the following points.

The middle-block period in Definition 11 is $P=M+x^2+1$.
The shift increment displayed in equation (19) is $M+x^2$ instead.
Identity \eqref{eq:canonical} uses $P$, derived from the definitions.
For example $(x,y,z)=(1,2,1)$ gives $M=4$, $P=6$ and
$G=\{0,1,3,4,9,10\}$. The correct $k=1$ shift is $\lambda=15$,
giving $\KA(1,2,2)$ of aperture 25. The displayed increment would give
$\lambda=14$, hence aperture 24 and a different set.

Appendix B of~\cite{RK2021} displays a first-hole exception for $y=1,x\geq1$.
At $(x,y,z)=(1,1,1)$, however, $G=\{0,1,2,5,6\}$ and
$G+G=[0,8]\cup\{10,11,12\}$: the first hole is 9.
The exception in Lemma~\ref{lem:hole} therefore uses $x\geq2$;
the $z=0$ case is handled separately.
These observations identify the version and definitions used in this proof.
They do not assert a publisher-issued corrigendum or priority for the observations.

\section{Computational record and reproduction}\label{app:computation}
Table~\ref{tab:certificate} records the full direct run used in
Proposition~\ref{prop:finite}. Counts refer to parameter occurrences
and shifts, not distinct arrays.
\begin{table}[ht]
\centering
\caption{Direct finite certificate through $N=329$.}\label{tab:certificate}
\begin{tabular}{lr}
\toprule
Quantity & Count\\
\midrule
Eligible generators & 246\,028\\
Unordered pairs processed, including diagonals & 7\,044\,825\,631\\
Feasible shifts visited & 2\,680\,100\,974\\
Shifts whose sensor counts lie in $[2,329]$ & 151\,726\,466\\
Exact equality-set comparisons completed & 1\,783\\
Largest generator aperture encountered & 27\,224\\
\bottomrule
\end{tabular}
\end{table}

Coordinates, loop totals and apertures use signed 64-bit integers.
Representation counts are at most $\card G\leq329$, hence fit in the
count type. The implementation visits unordered pairs, adding two for
an off-diagonal pair and one for a diagonal pair to obtain the ordered
counts used in \eqref{eq:finitecount}.
Generator buffers have 330 slots; the reflected union before
deduplication has at most 658 elements and uses 660 slots. Baseline
construction uses 990 slots, exceeding its bound $\card G+\card A\leq658$.
Allocation failure terminates execution, and assertions are enabled.
Independently of the observed maximum aperture, the coarse bounds
$x,y,z\leq329$, $M\leq13776$, $P\leq122018$ give $m<41$ million;
no coordinate calculation approaches the 64-bit limit.

The direct program includes \texttt{finite\_certificate.c} for shared
set construction, sorting, baseline storage and equality matching.
The formula-based program visits $28\,230\,844$ candidates and obtains
the same aperture and equality-count vectors. The exported witnesses
are checked by \path{check_witnesses.py}, which imports neither C
implementation. Polynomial identity checks using SymPy provide
additional exact algebraic corroboration.

The supplement's README gives standalone C11 commands for the direct run.
From its \texttt{verification} directory, the independent witness check
requires ordinary Python~3:
\begin{verbatim}
python audits/check_witnesses.py
\end{verbatim}
Its expected report has \texttt{status: PASS}, $1\,783$ occurrences,
$401$ distinct sets and agreement of all per-$N$ counts. The supplement's
source digest manifest binds the program and result files to the
verification snapshot; its index includes environment details and
commands for regenerating the witnesses.

The independent full-range Python check requires NumPy and uses no
floating-point convolution. SymPy is needed only for the auxiliary
symbolic checks. The supplement pins the tested versions and supplies
a driver that verifies the file manifest before and after execution:
\begin{verbatim}
python -m pip install -r requirements.txt
python run_checks.py --mode hashes
python run_checks.py --mode full --output /tmp/kloeve-full
\end{verbatim}
The \texttt{full} mode runs all finite certificates through $N=329$;
\texttt{quick} instead reruns enumeration only through $N=44$ and is
an installation check, not a replacement for Proposition~\ref{prop:finite}.
The full Python run agrees with the generator, shift and equality counts
in Table~\ref{tab:certificate}. It uses integer convolution, not pair loops.
A full direct C run under
AddressSanitizer and UndefinedBehaviorSanitizer produced identical
output. The supplement
documents implementation dependencies and the scope of each check.

\section*{Declarations}
The authors used Claude (Anthropic) and GPT (OpenAI) to assist with code
development, proof development and logical checks, literature searches,
and drafting portions of the text.
The authors assume responsibility for all content.


\end{document}